\documentclass[a4paper,12pt]{amsproc}
\usepackage{amsmath,amsfonts,amssymb,amsthm,anysize}
\usepackage{enumerate}
\usepackage{epsfig}
\usepackage{supertabular}
\usepackage{graphicx}
\usepackage{color}

\newcommand{\Cay}[2]{{\rm Cay}(#1,#2)}
\newcommand{\F}{\mathbb{F}_2}

\newcommand{\lla}{\langle}
\newcommand{\rr}{\rangle}

\newcommand{\bb}{\mathbb}

\usepackage{indentfirst}

\newtheorem{theorem}{Theorem}
\newtheorem{corollary}[theorem]{Corollary}
\newtheorem{lemma}[theorem]{Lemma}
\newtheorem{proposition}[theorem]{Proposition}
\newtheorem{definition}{Definition}
\newtheorem{example}{Example}[section] 
\newtheorem{remark}{Remark}[section]

\usepackage{url, hyperref}
\hypersetup{citecolor=blue, linkcolor=blue, colorlinks=true}

\begin{document}
%
\title{The BCH Family of Storage Codes on Triangle-Free Graphs is of Unit Rate}
%
%
%

\author{Haihua~Deng}
\address{ %
	Department of Mathematics\\
	Southern University of Science and Technology\\
	Shenzhen 518055, China}
\email{12131225@mail.sustech.edu.cn}

\author{Hexiang Huang}
\address{ %
	Department of Mathematics\\
	Southern University of Science and Technology\\
	Shenzhen 518055, China}
\email{hhxiang1999@foxmail.com}

\author{Guobiao Weng}
\address{ %
	School of Mathematical Sciences\\
	Dalian University of Technology\\
	Dalian 116024, China}
\email{weng@dlut.edu.cn}

\author{Qing Xiang$^*$}
\address{ %
	Department of Mathematics\\
	Southern University of Science and Technology\\
	Shenzhen 518055, China
}
\email{xiangq@sustech.edu.cn}
\thanks{$^*$Research partially supported by the National Natural Science Foundation of China Grant No. 12071206, 12131011, 12150710510, and the Sino-German Mobility Programme M-0157.}


\begin{abstract}
Let $\Gamma$ be a simple connected graph on $n$ vertices, and let $C$ be a code of length $n$ whose coordinates are indexed by the vertices of $\Gamma$. We say that $C$ is a \textit{storage code} on $\Gamma$ if for any codeword $c \in C$, one can recover the information on each coordinate of $c$ by accessing its neighbors in $\Gamma$. The main problem here is to construct high-rate storage codes on triangle-free graphs. In this paper, we solve an open problem posed by Barg and Z\'emor in 2022, showing that the BCH family of storage codes is of unit rate. Furthermore, we generalize the construction of the BCH family and obtain more storage codes of unit rate on triangle-free graphs. 
\end{abstract}

\maketitle

\vspace{-0.5cm}
\setcounter{tocdepth}{2}
\tableofcontents

%

\section{Introduction}

A simple graph, also called a strict graph, is an unweighted, undirected graph containing no loops or multiple edges. A simple graph is said to be {\it connected} if there is a path between each pair of its vertices.

Let $\Gamma$ be a simple connected graph on $n$ vertices, and let $C$ be a code of length $n$ whose coordinates are indexed by the vertices of $\Gamma$. We say that $C$ is a \textit{storage code} on $\Gamma$ if for any codeword $c \in C$, one can recover the information on each coordinate of $c$ by accessing its neighbors in $\Gamma$. In 2014,  Mazumdar \cite{Mazumdar2014StorageCO, Mazumdar2017StorageCA}, Shanmugam and Dimakis \cite{Shanmugam2014BoundingMU} introduced storage codes on graphs independently. The concept of storage codes on graphs was introduced, in a different way, by the authors of \cite{cameron2014guessing} and \cite{chris2011guessing}. Throughout this paper, we will only consider binary linear storage codes.

Suppose that $\Gamma$ has $n$ vertices, say $v_1,v_2,\ldots,v_n$. We define a storage code on $\Gamma$ in the following way: let $A(\Gamma)$ be the adjacency matrix of $\Gamma$ whose rows and columns are indexed by the vertices $v_1,v_2,\ldots,v_n$. Let $  H := A(\Gamma)+  I$ where $I$ is the $n\times n$ identity matrix, and let $C$ be the linear code over $\F$ with $H$ as a parity-check matrix. Given a codeword $c=(c_{v_1},c_{v_2},\ldots,c_{v_n})\in C$, we are able to recover any $v_i^{\rm th}$ entry of $c$ by accessing its neighbors since the $v_i^{\text{th}}$ row of $  H$ implies a linear equation, namely, $c_{v_i}=\sum_{v_j \in N(v_i)}c_{v_j}$, where $N(v_i)$ is the set of neighbors of $v_i$ in $\Gamma$. The rate of a linear storage code $C$, denoted by $R(C)$, is simply the ratio of its dimension to the dimension of the ambient space. If we have a family of storage codes $\{C_m\}$, where $m$ is a parameter, assuming that $\lim\limits_{m\to \infty}R(C_m)$ exists, then this limit is called the {\it rate} of the family. 

Constructing a family of storage codes of unit rate is easy: let $\Gamma_n$ be the  complete graph on $n$ vertices, and let $C_n$ be the binary linear code defined by the equation $\sum_{i=1}^nx_{v_i}=0$. Then $C_n$ is a storage code on $\Gamma$ with rate $1-1/n$; hence the family $\{C_n\}$ is of unit rate.

In the above example, the graph used to obtain the storage code of rate close to one is very dense (in fact, as dense as possible), and contains a large number of cliques. It is therefore natural to consider the question of the largest attainable rate of storage codes on graphs that contain no cliques $K_t \; (t\geq 3)$, i.e., triangle-free graphs.

Constructing storage codes of high rate on such graphs represents a challenge. A triangle-free graph with many edges does not necessarily give rise to a storage code of  high rate. To see this, consider the complete bipartite graph $K_{t,t}$ which is triangle-free and quite dense, and a storage code $C$ on it. Note that there are two independent vertex sets of $K_{t,t}$ and so for each vertex, we can recover the message on it from the ones in the other (vertex) independent set; hence $R(C)\leq 1/2$. In early studies \cite{chris2011guessing}, the authors had conjectured that for triangle-free graphs, $R = 1/2$ is the largest attainable rate value. Later on this conjecture was refuted in \cite{cameron2014guessing} by some sporadic examples.

Recently, the authors of \cite{barg2022high} constructed four infinite families of storage codes on triangle-free graphs. They used the \textit{Cayley graph method}: Let $S$ be a subset of $\F^r$ such that $0\notin S$ and the sum of any three distinct vectors in $S$ is nonzero. Then the resulting Cayley graph $\Gamma = \Cay{\F ^r}{S}$ is triangle-free. Let $ H := A(\Gamma) + I$ and $C$ be the binary linear code defined by using $ H$ as its parity-check matrix. Then we obtain a storage code $C$ on the triangle-free graph $\Gamma$. Using this method, a proper subset $S\subseteq \F^r$ will give rise to a triangle-free graph and a storage code on it. In their work, the Hamming family is of rate $3/4$ and the BCH family shows the record of rate $0.8196$. It remains unknown whether the BCH family can approach unit rate or not; this was left as an open problem in \cite{barg2022high}.

Subsequently, the authors of \cite{bargmsch} and the authors of \cite{xiang2023unitrate} presented the generalized Hamming family which could reach unit rate. In this paper, we solve the open problem about the BCH family, showing that the BCH family is indeed of unit rate. We also generalize the construction of the BCH family to obtain more storage codes on triangle-free graphs with rates approaching one. 

\section{The BCH family is of unit rate}

\subsection{An upper bound}
\label{Upperbound_section}

In 2022, Barg and Z\'emor~\cite{barg2022high} presented a new family, the so-called BCH family, which can reach the rate of $0.8196$. This data can be calculated by using a computer. It is left as an open problem whether the BCH family can reach unit rate or not. In order to understand why the BCH family exhibits a phenomenon of high rate, we use the \textit{polynomial method} to investigate the intrinsic algebraic structure of the BCH family. As a consequence, we give an upper bound for the rank of the parity-check matrix of the BCH family, which shows that the BCH family is indeed of unit rate. 

The BCH family of storage codes is constructed by using  the Cayley graph method. We first recall the definition of Cayley graphs.

\begin{definition}
 Let $G$ be a finite multiplicatively written group with identity element $e$, and let $S$ be a subset of $G$ such that $e\notin S$ and $S=S^{-1}$, where $S^{-1}=\left\{g^{-1}\ \vert\ g\in S\right\}$. The Cayley graph on $G$ with connection set $S$, denoted by $\Gamma=\Cay{G}{S}$,  is the graph with elements of $G$ as vertices, two vertices $g_1,g_2\in G$ are adjacent if and only if $g_1g_2^{-1}\in S$.
\end{definition}

Now we are going to construct the Cayley graphs of the BCH family. Let $q=2^m$ with $m\ge 1$ being an integer. The vertex set is given by $G=\mathbb{F}_q^2$ and the connection set is given by $S_m\backslash \{0\}$, where 
\begin{align*}
S_m:=\left\{(a,a^3)\mid a\in \bb F_q\right\}\subseteq \bb F_q^2.
\end{align*}  
The graph is $\Gamma=\Cay{\bb F_q^2}{S_m\backslash \{0\}}$.

Let $H_m:=A(\Gamma)+I$ and $C_m$ be the binary linear code defined by using $H_m$ as a parity-check matrix. Since each row of the parity-check matrix $H_m$ for the storage code $C_m$ can be regarded as a characteristic vector of a coset in $\{x+S_m\mid x\in \bb F_q^2 \}$, we may call $H_m$ the \textit{coset matrix} of $S_m$ in $\bb F_q^2$. In order to better understand the structure of the matrix $H_m$, we will express the $(x,y)$-entry of $H_m$ as the value of a polynomial evaluated at $(x,y)$. More precisely, the coset matrix $H_m$ over $\mathbb{F}_2$ can be formulated as 
\begin{align*}
H_m&=(a_{x,y})_{x,y\in \bb F_q^2},    
\end{align*}
where the $(x,y)$-entry is given by 
\begin{align*}
a_{x,y}=\begin{cases}
 1, & \text{if } x-y\in S_m, \\
 0, & \text{otherwise. } 
\end{cases}    
\end{align*}

Next, we apply the \textit{polynomial method} to investigate the rank of $ H_m$. If we write $x=(x_1,x_2),y=(y_1,y_2)\in \bb F_q^2$, then $a_{x,y}$ can be expressed as the value of a polynomial $g$ evaluated at $(x,y)$:
\begin{equation*}
    \begin{aligned}
        a_{x,y}&=\left((x_1-y_1)^3-(x_2-y_2)\right)^{q-1}+1 \\
        &=\left(x_1^3+x_1^2y_1+x_1y_1^2+y_1^3+x_2+y_2\right)^{q-1}+1\\
        &=:g(x_1,x_2,y_1,y_2).
    \end{aligned}
\end{equation*}

\noindent Let $W_m=(a_{x,y}+1)_{x,y\in\mathbb{F}_q^2}$. Then $  W_m=  H_m+  J$, where $ J$ is the all-one matrix, and so 
\begin{align*}
\mathrm{rank}(H_m)-\mathrm{rank}(J)\leq \mathrm{rank}(W_m)\leq \mathrm{rank}(H_m)+\mathrm{rank}(J),    
\end{align*} 
that is,
\begin{align*}
\mathrm{rank}(  H_m)-1\leq \mathrm{rank}(  W_m)\leq \mathrm{rank}(  H_m)+1.
\end{align*}

Therefore, the matrix $W_m$ has almost the same rank as that of $H_m$. We define the {\it rate} of a square matrix $ A_{n\times n}$ to be the ratio of $\mathrm{rank}( A)$ to the size $n$; that is, $R( A)=\mathrm{rank}( A)/n$. That the BCH family is of unit rate is equivalent to saying that the rate of $W_m$ converges to $0$ as $m\rightarrow\infty$. Now the problem is reduced to computing the rank of $W_m$ whose entry $a_{x,y}+1$ is given by 
\begin{align*}
h(x_1,x_2,y_1,y_2)=(x_1^3+x_1^2y_1+x_1y_1^2+y_1^3+x_2+y_2)^{q-1}.    
\end{align*}

The following proposition simplifies the question further by dropping some terms from the polynomial $h$.

\begin{proposition}
    Let $  D_m=\left(f(x,y)\right)_{x,y\in \bb F_q^2}$, where 
    \[f(x_1,x_2,y_1,y_2)=(x_1^2y_1+x_1y_1^2+x_2+y_2)^{q-1}.\]
    Then $D_m$ has the same $\F$-rank as that of $W_m$.
\end{proposition}

\begin{proof}
Note that changing $(x_1,x_2)$ to $(x_1,x_2+x_1^3)$ is a permutation on $\bb F_q^2$. Thus changing $h(x_1,x_2,y_1,y_2)$ to $f=h(x_1,x_2+x_1^3,y_1,y_2+y_1^3)$ is in fact a permutation on the rows and columns of $W_m$. The conclusion of the proposition now follows.
\end{proof}

To find an upper bound on the rank of $D_m$, we first decompose $D_m$ as the product of two matrices. Let $$\Omega:=\left\{(l_1,l_2,l_3,l_4)\mid \sum_{i=1}^4l_i=q-1,0\leq l_i\leq q-1,\forall i\right\}.$$ Then we can expand the polynomial $f$ as follows:
\begin{align*}
    f&=(x_1^2y_1+x_1y_1^2+x_2+y_2)^{q-1} \\
    &=\sum_{(l_1,l_2,l_3,l_4)\in \Omega}{q-1\choose l_1,l_2,l_3,l_4}x_1^{2l_1+l_2}x_2^{l_3}y_1^{l_1+2l_2}y_2^{l_4} \\
    &=\left[\begin{array}{ccc}
         \cdots&{q-1\choose l_1,l_2,l_3,l_4}x_1^{2l_1+l_2}x_2^{l_3}&\cdots   
    \end{array}\right]
    \left[\begin{array}{c}
         \vdots  \\
          y_1^{l_1+2l_2}y_2^{l_4} \\
          \vdots
    \end{array}\right],
\end{align*}
where the coordinates of the row/column vector are indexed by elements in $\Omega$. Therefore we can write $  D_m$ as the product of two matrices
\[  D_m={ L R}=\left[\begin{array}{ccc}
         \cdots&{q-1\choose l_1,l_2,l_3,l_4}x_1^{2l_1+l_2}x_2^{l_3}&\cdots   
    \end{array}\right]\left[\begin{array}{c}
         \vdots  \\
          y_1^{l_1+2l_2}y_2^{l_4} \\
          \vdots
    \end{array}\right],\]
where the rows of $ L$ and columns of $ R$ are indexed by elements of $\bb F_q^2$. Let $N_m$ be the number of distinct nonzero monomials in $L$. That is,
\begin{align*}
    N_m:&= \#\left\{(2l_1+l_2,l_3)\,\bigg|\, {q-1\choose l_1,l_2,l_3,l_4}\equiv 1\pmod{2}\right\}.
\end{align*}
We then have an upper bound on $\mathrm{rank}(  D_m)$:
\begin{equation}
\label{upper_bound_for_Dm}
    \mathrm{rank}(D_m)\leq {\rm rank}( L)\leq N_m,
\end{equation}

Using some counting techniques, we can give an explicit formula for $N_m$; and hence obtain an upper bound on the rank of $D_m$. The obtained upper bound is good enough for us to show that the BCH family is of unit rate. We state the following theorem whose proof is postponed to the next subsection.

\begin{theorem}
\label{upper_bound_theorem}
Let $D_m$ be defined as above with $m\ge 1$ being an integer. Then
\begin{align*}
\mathrm{rank}(D_m)\leq \frac{1+\sqrt{2}}{2}(2+\sqrt{2})^m,
\end{align*}
and so
\begin{align*}
R(D_m)\leq \frac{1+\sqrt{2}}{2}\left(\frac{2+\sqrt{2}}{4}\right)^{m}.
\end{align*}
\end{theorem}

\subsection{Proof of Theorem \ref{upper_bound_theorem}}
\label{details_section}

Since the sequence of numbers $N_m$ is defined by a property involving multinomial coefficients, we will use Lucas' theorem to analyse the behavior of $N_m$. Surprisingly, we can even compute the exact values of $N_m$.

Let $n$ be a non-negative integer and $p$ a prime. Suppose that the base $p$ expansion of $n$ is given by $n=n_kp^k+n_{k-1}p^{k-1}+\cdots +n_1p+n_0$, where $0\leq n_i\leq p-1$ for all $i$. We may use the abbreviation $n=\lla n_kn_{k-1}\cdots n_1n_0\rr_p$ or $n=\lla n_k,n_{k-1},\cdots, n_1,n_0\rr_p$. In the case where $p=2$, we may drop the subscript $p$. We state Lucas' theorem as follows.
\begin{theorem}[Lucas' Theorem~\cite{MR938818}]
\label{Lucas_theorem}
Let $p$ be a prime, and express the non-negative integers $n,l_1,l_2,\ldots,l_s$ in base $p$ as
\begin{align*}
n=\lla n_k,n_{k-1},\ldots, n_1,n_0\rr_p;\quad l_i=\lla l_{i,k},l_{i,k-1},\ldots ,l_{i,1},l_{i,0}\rr_p,
\end{align*}
where $n_j,l_{i,j}\in\{0,1,\ldots,p-1\}$ for $j=0,1,\ldots,k$ and $i=1,2,\ldots,s$. Then
\[{n\choose l_1,l_2,\ldots,l_s}\equiv \prod_{j=0}^k{n_j\choose l_{1,j},l_{2,j}\ldots,l_{s,j}}\pmod{p}.\]
\end{theorem}

In the case where $p=2$, we will drop (mod $2$) to simplify notation. Before doing the actual computations, we will fix some notation as follows.

\begin{definition}
    Let $a,b,c$ be non-negative integers. We write $a+b\lessdot c$, if the following conditions hold:
    \[a_i+b_i\leq c_i\text{ for all }i=0,\ldots,k,\]
    where $a=\lla a_ka_{k-1}\cdots a_1a_0\rr_2,\;b=\lla b_kb_{k-1}\cdots b_1b_0\rr_2,\;c=\lla c_kc_{k-1}\cdots c_1c_0\rr_2$.
\end{definition}

For $0\leq s\leq q-1$, define 
\begin{align*}
B_s:=\left\{2l_1+l_2\,\big|\, {q-1\choose l_1,l_2,q-1-s,l_4}\equiv 1\text{ for some }l_4\right\}
\end{align*} 
and $b_s:=|B_s|$. Note that the base $2$ expansion of $q-1$ is $\lla \underbrace{11\cdots 1}_m\rr$. By Theorem \ref{Lucas_theorem} we know that ${q-1\choose l_1,l_2,q-1-s,l_4}\equiv 1\pmod{2}$ if and only if the addition $l_1+l_2+(q-1-s)+l_4=q-1$ involves no carries, which in turn is equivalent to  $l_1+l_2\lessdot s$ and $l_4=s-l_1-l_2$. We now rewrite $B_s$ as
\begin{equation*}
B_s=\left\{2l_1+l_2\mid l_1+l_2\lessdot s\right\}.
\end{equation*}
Note that we have $N_m=\sum_{s=0}^{q-1}b_s$. 

\begin{lemma}
\label{divide_binary_series_to_disjoint_parts}
    Let $s=\lla\alpha_1,\alpha_2,\ldots,\alpha_n,\beta_1,\beta_2,\ldots,\beta_k\rr$. Then \[B_s=B_{s_1}\times2^k+B_{s_2}:=\left\{r2^k+t\,\big|\, r\in B_{s_1},t\in B_{s_2}\right\},\]
    where $s_1=\lla \alpha_1,\alpha_2,\ldots,\alpha_n\rr$ and $s_2=\lla \beta_1,\beta_2,\ldots,\beta_k\rr$.
\end{lemma}

\begin{proof}
    On the one hand, $B_s\subseteq B_{s_1}\times2^k+B_{s_2}$. This can be seen as follows. Assume that $l_1+l_2\lessdot s$. By the division algorithm we write $l_1=r_1 2^k+t_1,0\leq t_1<2^k$, where the quotient and remainder, $r_1,t_1$, are uniquely determined. Similarly for $l_2$ we obtain the quotient and the remainder, $r_2,t_2$, respectively. As $l_1+l_2\lessdot s$, we have $r_1+r_2\lessdot s_1,t_1+t_2\lessdot s_2$ and thus $2l_1+l_2=2(r_1\times2^k+t_1)+(r_2\times2^k+t_2)=(2r_1+r_2)\times2^k+(2t_1+t_2)\in B_{s_1}\times2^k+B_{s_2}$.

    On the other hand, $B_s\supseteq B_{s_1}\times2^k+B_{s_2}$: Assume $r_1+r_2\lessdot s_1,t_1+t_2\lessdot s_2$. Let $l_1=r_1\times 2^k+t_1,l_2=r_2\times 2^k+t_2$. Then $l_1+l_2\lessdot s$. So $(2r_1+r_2)\times2^k+(2t_1+t_2)=2(r_1\times2^k+t_1)+(r_2\times2^k+t_2)=2l_1+l_2\in B_s$.
\end{proof}

\begin{proposition}
\label{allone_case}
    Let $i$ be a positive integer. Then 
    \[b_{2^{i-1}-1}=2^i-1.\]
\end{proposition}

\begin{proof}
    Note that $B_{2^{i-1}-1}=\{2l_1+l_2\mid l_1+l_2\lessdot 2^{i-1}-1\}$. Fixing $l_1=0$, we can take $l_2=0,1,2,\ldots,2^{i-1}-1$, then $2l_1+l_2=0,1,2,\ldots,2^{i-1}-1$.

    Let $l_1+l_2=2^{i-1}-1$. Then we have $l_1+l_2\lessdot 2^{i-1}-1$ and $2l_1+l_2=l_1+2^{i-1}-1$. As $l_1$ varies from $0$ to $2^{i-1}-1$, $2l_1+l_2$ varies from $2^{i-1}-1$ to $2^i-2$. So $B_{2^{i-1}-1}=\left\{0,1,2,\ldots,2^i-2\right\}$. The claim now follows.
\end{proof}

\begin{lemma}
\label{divide_segment}
    Let $s=\lla\alpha_1,\alpha_2,\ldots,\alpha_n,0,\beta_1,\beta_2,\ldots,\beta_k\rr$. Then
    \[b_s=b_{s_1}b_{s_2},\]
    where $s_1=\lla \alpha_1,\alpha_2,\ldots,\alpha_n\rr$ and $s_2=\lla \beta_1,\beta_2,\ldots,\beta_k\rr$.
\end{lemma}

\begin{proof}
    By Lemma \ref{divide_binary_series_to_disjoint_parts} we have
    \[B_s=B_{s_1}\times2^{k+1}+B_{s_2}.\]
    Note that for any $t\in B_{s_2}$, $t\leq 2s_2<2^{k+1}$. Assume there are two pairs $(r_1,t_1),(r_2,t_2)\in B_{s_1}\times B_{s_2}$, such that $r_1\times2^{k+1}+t_1=r_2\times2^{k+1}+t_2$. Then $(r_1-r_2)\times2^{k+1}+(t_1-t_2)=0$ and thus $r_1-r_2=t_1-t_2=0$, i.e., $(r_1,t_1)=(r_2,t_2)$. Hence $\#B_s=\#(B_{s_1}\times2^{k+1}+B_{s_2})=\#\left(B_{s_1}\times B_{s_2}\right)$.
\end{proof}

\begin{example}
\label{initial_term}
    By direct calculations, we have
    \begin{align*}
        b_0&=b_{2^0-1}=2^1-1=1, &\text{by Proposition \ref{allone_case}} \\
        b_1&=b_{2^1-1}=2^2-1=3, &\text{by Proposition \ref{allone_case}} \\
        b_2&=b_{\lla 10\rr}=b_1b_0=3, &\text{by Lemma \ref{divide_segment}} \\
        b_3&=b_{2^2-1}=2^3-1=7. &\text{by Proposition \ref{allone_case}}
    \end{align*}
    Thus $N_1=b_0+b_1=4,N_2=b_0+b_1+b_2+b_3=14$. We will use the initial values $N_1,N_2$ to determine the general formula of $N_m$ in Theorem \ref{general_term}.
\end{example}

\begin{proposition}
\label{recusion_formula}
    The sequence of numbers $N_m$ satisfies:
    \begin{equation*}
    \begin{aligned}
        N_m=\sum_{j=1}^{m+1}(2^j-1)N_{m-j},\quad m\geq 1,
    \end{aligned}
    \end{equation*}
    where $N_0=N_{-1}=1$.   
\end{proposition}

\begin{proof}
     Define
     \begin{equation*}
     \begin{aligned}
         E^{(m)}:&=\left\{t\,\bigg|\, 0\leq t\leq 2^m-1\right\}\\
         &=\left\{\lla t_{m-1},\ldots,t_1,t_0\rr\,\bigg|\, t_i\in \{0,1\},\forall i\right\}, \\
         E_j^{(m)}:&=\left\{\lla\underbrace{1,1,\ldots,1}_{j-1},0,t_{m-j-1},\ldots,t_1,t_0\rr\,\bigg|\, t_i\in \{0,1\},\forall i\right\} \\
         &=\{2^{j-1}-1\}\times2^{m-j+1}+E^{(m-j)}, j=1,\ldots,m-1; \\
         E_{m}^{(m)}:&=\left\{\lla\underbrace{11\cdots1}_{m-1}0\rr\right\},\quad
         E_{m+1}^{(m)}:=\left\{\lla\underbrace{11\cdots1}_{m}\rr\right\}.
     \end{aligned}
     \end{equation*} 
     It is clear that $E^{(m)}$ is the disjoint union of $E_j^{(m)},j=1,\ldots,m+1$, namely $E^{(m)}=\cup_{j=1}^{m+1}E_j^{(m)}$.  Note that $N_m=\sum_{t=0}^{2^m-1}b_s=\sum_{s\in E^{(m)}}b_s$. Then by Proposition \ref{allone_case} and Lemma \ref{divide_segment} we have
     \begin{equation*}
     \begin{aligned}
         N_m&=\sum_{s\in E^{(m)}}b_s=\sum_{j=1}^{m-1}\sum_{s\in E_j^{(m)}}b_s+2^{m}-1+2^{m+1}-1  \\
         &=\sum_{j=1}^{m-1}\sum_{t\in E^{(m-j)}}b_{2^{j-1}-1}b_t+2^{m}-1+2^{m+1}-1  \\
         &=\sum_{j=1}^{m-1} b_{2^{j-1}-1}\sum_{t\in E^{(m-j)}}b_t+2^{m}-1+2^{m+1}-1 \\
         &=\sum_{j=1}^{m-1}b_{2^{j-1}-1}N_{m-j}+2^{m}-1+2^{m+1}-1  \\
         &=\sum_{j=1}^{m-1}(2^j-1)N_{m-j}+2^{m}-1+2^{m+1}-1 \\
         &=\sum_{j=1}^{m+1}(2^j-1)N_{m-j}.
     \end{aligned}
     \end{equation*}
\end{proof}

\begin{theorem}
\label{general_term}
We have   
\begin{equation}
\label{formula_of_Mi}
N_m=\frac{1+\sqrt{2}}{2}(2+\sqrt{2})^m+\frac{1-\sqrt{2}}{2}(2-\sqrt{2})^m,\quad m\geq 0.
\end{equation}
\end{theorem}

\begin{proof}
By Proposition \ref{recusion_formula}, we have
\begin{equation}
\label{formula_4}
N_m=\sum_{j=1}^{m+1}(2^j-1)N_{m-j},\quad m\geq 1.
\end{equation}
Replacing $m$ by $m-1$, we get
\begin{equation}
\label{formula_5}
\begin{aligned}   
N_{m-1}&=\sum_{j=1}^{m}(2^j-1)N_{m-1-j}\\
&=\sum_{j=2}^{m+1}(2^{j-1}-1)N_{m-j},\quad m\geq 2, 
\end{aligned}
\end{equation}
Using \eqref{formula_4} and \eqref{formula_5}, we obtain
\begin{equation*}
N_m-2N_{m-1}=N_{m-1}+\sum_{j=2}^{m+1}N_{m-j},\quad m\geq 2.\quad 
\end{equation*}
It follows that
\begin{equation} 
\label{formula_6}
N_m=3N_{m-1}+\sum_{j=2}^{m+1}N_{m-j},\quad m\geq 2.
\end{equation}
Again, replacing $m$ by $m-1$, we get
\begin{equation}
\label{formula_7}
N_{m-1}=3N_{m-2}+\sum_{j=3}^{m+1}N_{m-j},\quad m\geq 3, 
\end{equation}
Using \eqref{formula_6} and \eqref{formula_7}, we obtain
\begin{equation*}
N_m-N_{m-1}=3N_{m-1}-3N_{m-2}+N_{m-2},\quad m\geq 3.\quad
\end{equation*}
Conseqeuntly,
\begin{equation*}
N_m=4N_{m-1}-2N_{m-2},\quad m\geq 3.
\end{equation*}
Taking $m=2$, we find that the initial values $N_0=1,N_1=4,N_2=14$ satisfy this linear recurrence relation. So the above recurrence holds whenever $m\geq 2$. Solving the linear recurrence we obtain 
\begin{align*}
N_m=\frac{1+\sqrt{2}}{2}(2+\sqrt{2})^m+\frac{1-\sqrt{2}}{2}(2-\sqrt{2})^m,\quad m\geq 2.    
\end{align*}
\end{proof}

We are now ready to give the proof of Theorem \ref{upper_bound_theorem}.

\noindent\textit{Proof of Theorem \ref{upper_bound_theorem}.}
    We have the following upper bound:
    \begin{equation*}
    \begin{aligned}
         \mathrm{rank}(  D_m)&\leq N_m 
         =\frac{1+\sqrt{2}}{2}(2+\sqrt{2})^m+\frac{1-\sqrt{2}}{2}(2-\sqrt{2})^m \\
         &\leq \frac{1+\sqrt{2}}{2}(2+\sqrt{2})^m. 
    \end{aligned}
    \end{equation*}
    Therefore,
    \begin{equation*}
    \begin{aligned}
        R(  D_m)&=\frac{\mathrm{rank}(  D_m)}{q^2} 
        \leq\frac{1+\sqrt{2}}{2}\left(\frac{2+\sqrt{2}}{4}\right)^m.
    \end{aligned}
    \end{equation*}
From the above upper bound on $R(D_m)$, we immediately see that the BCH family is of unit rate. $\hfill\square$

\subsection{The ambient graphs of the BCH Family}

For a given positive integer $m$, the graph of the BCH family is $\Gamma(V,E)=\Cay{\bb F_q^2}{S_m\backslash\{0\}}$, where $q=2^m$. The number of vertices is $N=|V|=q^2=2^{2m}$. Note that $\Gamma$ is a regular graph and each vertex has degree $|S_m|-1=q-1=2^m-1$, so the number of edges is 
\[|E|=\frac{N(2^m-1)}{2}=\frac{N(\sqrt{N}-1)}{2}=O(N^{3/2}).\]

It is clear that $\Gamma$ is simple.  We claim that $\Gamma$ is connected when $m>2$. The proof is given in the next section.

We show that $\Gamma$ is triangle-free: Let $a,b,c\in \bb F_q$ be distinct nonzero elements such that $a+b+c=0$. We claim that $a^3+b^3+c^3\neq 0$. If not, then $c^3=(a+b)^3=a^3+a^2b+ab^2+b^3=a^3+b^3$ and we obtain $a=b$, a contradiction. Hence the sum of any three distinct nonzero vectors in $S_m$ is nonzero and thus $\Gamma$ is triangle-free.

\section{The generalized BCH family}
\label{section_generalized_BCH_family}

Recall that in the BCH family, we investigate the coset matrix of $S_m$ in $\bb F_q^2$, where $q=2^m$ and $S_{m}=\{(a,a^3)\mid a\in \bb F_q\}$. Now define 
\[S_{n,m}:=\{(a,a^n)\mid a\in \bb F_q\}\subseteq \bb F_q^2,\]
where $n$ is a fixed odd integer and $1< n\leq q-1$. Then we obtain the \textit{generalized BCH family}  $F_n$ on the graph $\Gamma_{n,m}=\Cay{\bb F_q^2}{S_{n,m}\backslash\{0\}}$.

\begin{remark}
    In the above generalization, we require $n$ to be odd. In fact, the matrix $H_{n,m}$ has the same rank as $H_{n/2,m}$ when $n$ is even, where $H_{n,m}$ denotes the coset matrix of $S_{n,m}$ in $\bb F_q^2$.
\end{remark}

 To prove that $\Gamma_{n,m}$ is connected, we need to show any vector in $\bb F_q^2$ is a sum of vectors in $S_m$. That is, viewing $\bb F_q^2$ as a $2m$-dimensional $\F$-vector space, we need to show that $S_m$ contains a basis of $\bb F_q^2$.

We now show that when $m$ is large enough, the graph $\Gamma_{n,m}$ is connected. The following proof can be found in most coding theory textbooks. For more details, we refer the readers to \cite{ling_xing_2004}.

\begin{theorem}
 Let $n>1$ be an odd integer. If $2^{\frac{m}{2}}+1>n$, then $S_m$ contains a $\F$-basis for $\bb F_q^2$; and the graph $\Gamma_{n,m}$ is connected.
\end{theorem}

\begin{proof}
    Let $\bb F_q^*=\lla \alpha\rr$. We claim that $\{(\alpha^k,\alpha^{nk})\mid k=0,1,\ldots,2m-1\}$ are linearly independent over $\F$. Assume that $\sum_{k=0}^{2m-1}c_k\alpha^k=\sum_{k=0}^{2m-1}c_k\alpha^{nk}=0$, where $c_k\in \F$. Let $g(x)=\sum_{k=0}^{2m-1}c_kx^k$. Then $g(\alpha)=g(\alpha^n)=0$.
    
    Let $p_1(x),p_2(x)$ be the minimal polynomials of $\alpha,\alpha^n$ in $\F[x]$ respectively. We know that $p_1,p_2$ are irreducible polynomials. As $\alpha$ is a primitive element of $\bb F_q$, we have $\deg p_1=m$. Note that $\alpha$ and $\alpha^n$ are not conjugate to each other as $n$ is odd and $1<n< 2^{\frac{m}{2}}+1<q-1$, so $p_1,p_2$ are coprime to each other. 

    Suppose that $\deg p_2=d$. We show that $d=m$: We know that $d|m$ and $\alpha^n=\alpha^{n2^d}$, so $(2^m-1)|n(2^d-1)$ and thus $n\geq \frac{2^m-1}{2^d-1}$. Combining with the assumption that $2^{\frac{m}{2}}+1>n$, we have $2^{\frac{m}{2}}+1>\frac{2^m-1}{2^d-1}$, so $d>\frac{m}{2}$ and consequently $d=m$ since $d|m$.

    The polynomial $g(x)$ should be a multiple of $p_1(x)p_2(x)$ since $g(x)$ has the roots $\alpha,\alpha^n$. As the degree of $p_1(x)p_2(x)$ is $2m$ and $g(x)$ cannot be of degree $2m$, we deduce that $g(x)$ is the zero polynomial. Hence the claim follows. 
\end{proof}

\begin{corollary}
    Let $r$ be a positive integer and $n=2^r+1$. Then the graph $\Gamma_{n,m}$ is connected if $m>2r$.
\end{corollary}

For what values of $n$, will $\Gamma_{n,m}$ be triangle-free? The next lemma answers this question.

\begin{lemma}
    The graph $\Gamma_{n,m}$ is not triangle-free if and only if the equation $(x+1)^n+1=x^n+1$ has solution other than $x=0,1$.
\end{lemma}

\begin{proof}
    First note that $\Gamma_{n,m}$ is triangle-free if and only if the sum of any three nonzero distinct vectors in $S_{n,m}$ is nonzero. Assume that there are three distinct elements $a,b,c\in \bb F_q^*$ such that $a+b+c=a^n+b^n+c^n=0$. Then $c^n=(a+b)^n=a^n+b^n$. Let $x=ab^{-1}$. We obtain $(x+1)^n=x^n+1$, where $x\neq 0,1$. Hence the necessary condition holds.

    To verify the sufficient condition, we assume that there exists an $x\neq 0,1$ satisfying $(x+1)^n+1=x^n+1$. Then $a=x,b=1,c=x+1$ are three distinct nonzero elements. The vectors $(a,a^n),(b,b^n),(c,c^n)$ will cause a triangle in the graph as the sum of them is zero.
\end{proof}

\begin{proposition}
    Let $n=2^r+1$, where $r\geq 1$ is an integer. Then the graph $\Gamma_{n,m}$ is triangle-free if and only if $\gcd(r,m)=1$.
\end{proposition}

\begin{proof}
    We only need to consider the equation
    \begin{align*}
        (x+1)^{2^r+1}&=x^{2^r+1}+1,\\
        x(x^{2^r-1}+1)&=0.
    \end{align*}

    The equation $x(x^{2^r-1}+1)=0$ only has solutions $x=0,1$ if and only if $\gcd(x(x^{2^r-1}+1),x^q+x)=x(x+1)$. That is
    \begin{align*}
        x(x+1)&=\gcd(x(x^{2^r-1}+1),x^q+x)\\
        &=x\gcd(x^{2^r-1}+1,x^{2^m-1}+1)\\
        &=x(x^{2^{\gcd(r,m)}-1}+1),
    \end{align*}
    which means that $\gcd(r,m)=1$.
\end{proof}

\begin{proposition}
   Let $n=2^r-1$, where $r\geq 2$ is an integer. Then the graph $\Gamma_{n,m}$ is triangle-free if and only if $\gcd(r-1,m)=1$.
\end{proposition}

\begin{proof}
    We only need to consider the equation
    \begin{align*}
        (x+1)^{2^r-1}&=x^{2^r-1}+1,\\
        \sum_{k=1}^{2^r-2}x^k&=x\frac{1+x^{2^r-2}}{1+x}=0.
    \end{align*}

    The equation only has solutions $x=0,1$ if and only if $\gcd(x^{2^r-2}+1,x^{q-1}+1)=x+1$, which implies that $\gcd(2^r-2,2^m-1)=\gcd(2^{r-1}-1,2^m-1)=2^{\gcd(r-1,m)}-1=1$ and thus $\gcd(r-1,m)=1$. The converse is also true.
\end{proof}


%

\section{The case when $n=2^r+1$}

In this section, we are going to show that the generalized BCH family $F_n$ is of unit rate provided that $n=2^r+1$, where $r$ is a positive integer. 

Let $H_m$ be the coset matrix of $S_m$ in $\bb F_q^2$. Then we can formulate $H_m$ as 
\begin{align*}
    H_m=(g(x,y))_{x,y\in \bb F_q^2}, 
\end{align*}
where the $(x,y)$-entry is given by 
\begin{align*}  
    g(x,y)=\left((x_1+y_1)^{2^r+1}+x_2+y_2\right)^{q-1}+1.
\end{align*}

\noindent By the same argument as in Section \ref{Upperbound_section}, we deduce that $H_m$ has almost the same rank as that of the matrix $D_m:=\left(f(x,y)\right)_{x,y\in \bb F_q^2}$, where the $(x,y)$-entry is given by
\begin{align*}  
f(x,y)=\left(x_1^{2^r}y_1+x_1y_1^{2^r}+x_2+y_2\right)^{q-1}.
\end{align*}

\noindent Similarly, we have an upper bound:  $\mathrm{rank}(D_m)\leq N_m$, where
\begin{align*}
    N_m:&=\#\left\{(2^rl_1+l_2,l_3)\,\Bigg|\, {q-1\choose l_1,l_2,l_3,l_4}\equiv 1\right\} \\
    &=\#\left\{(2^rl_1+l_2,l_3)\mid l_1+l_2\lessdot q-1-l_3\right\}.
\end{align*}

\noindent We may redefine $B_s$ in Section \ref{details_section} by 
$$B_s:=\left\{2^rl_1+l_2\mid l_1+l_2\lessdot s\right\},$$ 
where $0\leq s\leq q-1$. Then
\begin{align*}
    N_m&=\sum_{s=0}^{2^m-1}\#\left\{2^rl_1+l_2\mid l_1+l_2\lessdot q-1-s\right\} \\
    &=\sum_{s=0}^{2^m-1}|B_{q-1-s}|=\sum_{s=0}^{2^m-1}|B_s|.
\end{align*}

\noindent For this more general definition of $B_s$, Lemma \ref{divide_binary_series_to_disjoint_parts} still holds. We omit the proof since it is completely the same as before.

\begin{lemma}
    Let $r$ be a positive integer and $B_s$ defined above. Then
    \begin{enumerate}
    \vspace{+0.2cm}
        \item $\sum_{s=0}^{2^k-1}|\{(l_1,l_2)\mid l_1+l_2\lessdot s \}|=4^k$;\vspace{+0.2cm} \label{basic_count}
        \item $N_k=\sum_{s=0}^{2^k-1}|B_s|=4^k$, for $k=0,1,\ldots,r$;\vspace{+0.2cm}
        \item $N_{r+1}=\sum_{s=0}^{2^{r+1}-1}|B_s|\le 15\times 4^{r-1}$.
    \end{enumerate}
\end{lemma}

\begin{proof}
1) We can classify all $s$ between $0$ and $2^k-1$ by its weight, namely the number of $1$s in its base $2$ expansion. If the weight of $s$ is $i$, then $\#\{(l_1,l_2)\mid l_1+l_2\lessdot s\}=3^i$ since each pair $(l_1,l_2)$ implies a distribution of each $1$s to $l_1$,$l_2$ or $l_3$, where $l_3=s-l_1-l_2$. We have
    \[\sum_{s=0}^{2^k-1}\#\{(l_1,l_2)\mid l_1+l_2\lessdot s \}=\sum_{i=0}^k{k \choose i}3^i=4^k.\]
    
2) Assume $k\leq r$. For each $0\leq s\leq 2^k-1$, we have a map from $\{(l_1,l_2)\mid l_1+l_2\lessdot s \}$ to $B_s$, sending $(l_1,l_2)$ to $2^rl_1+l_2$. We want to show this is bijective map. It is clearly surjective, so we only need to show it is injective.

If there are two pairs $(l_1,l_2),(l_1',l_2')$ such that $2^rl_1+l_2=2^rl_1'+l_2'$, then $2^r(l_1-l_1')+l_2-l_2'=0$ and thus $l_1-l_1'=l_2-l_2'=0$ as $0\leq l_1,l_2,l_1',l_2'\leq s\leq 2^r-1$. Hence the map is injective and thus bijective. The result follows by \eqref{basic_count}.

3) Notice that $B_1=\{2^rl_1+l_2\mid l_1+l_2\lessdot 1\}=\{0,1,2^r\}$. We first calculate 
\begin{align*}
&\sum_{s=0}^{2^r-1}|B_{2s+1}|
=\sum_{s=0}^{2^r-1}\#(B_s\times 2+B_1) \\
        =&\sum_{s=0}^{2^r-1}\#(B_s\times 2+\{0,1,2^r\}) \\
        =&\sum_{s=0}^{2^r-1}\#\left(B_s\times2+1\right)+\sum_{s=0}^{2^r-1}\#\left(B_s\times2+2^r\right)\\
        &+\sum_{s=0}^{2^r-1}\#\left(B_s\times2\right)-\sum_{s=0}^{2^r-1}\#\left[\left(B_s\times2\right)\cap \left(B_s\times2+2^r\right)\right] \\
        =&3N_r-\sum_{s=0}^{2^r-1}\#\left[\left(B_s\times2\right)\cap \left(B_s\times2+2^r\right)\right] \\
        =&3N_r-\sum_{s=0}^{2^r-1}\#\left[B_s\cap \left(B_s+2^{r-1}\right)\right]:=3N_r-N'.
\end{align*}
    
To determine the second term $N'$, let $(l_1,l_2),(l_1',l_2')$ with $l_1+l_2\lessdot s,l_1'+l_2'\lessdot s$, where $0\leq s\leq 2^r-1$.  Assume $2^rl_1+l_2=2^rl_1'+l_2'+2^{r-1}$ in the intersection. Then we obtain $2^r(l_1-l_1')+(l_2-l_2')=2^{r-1}$. This equality holds if $l_1-l_1'=0,l_2-l_2'=2^{r-1}$. Hence $2^rl_1'+l_2'+2^{r-1}$ is in the intersection if $2^{r-1}\le s\le 2^r-1$ and $l_1'+l_2'\lessdot s-2^{r-1}$. Then 
    \[N'\ge\sum_{s=2^{r-1}}^{2^r-1}|B_{s-2^{r-1}}|
        =\sum_{s=0}^{2^{r-1}-1}|B_s|=4^{r-1}.\]

The value of $N_{r+1}$ is given by 
\begin{align*}       N_{r+1}=&\sum_{s=0}^{2^{r+1}-1}|B_s|=\sum_{s=0}^{2^r-1}(|B_{2s}|+|B_{2s+1}|) \\
&=\sum_{s=0}^{2^r-1}\#(B_s\times 2)+\sum_{s=0}^{2^r-1}|B_{2s+1}| \\
&=N_r+3N_r-N'\le4^{r+1}-4^{r-1}=15\times 4^{r-1}.
\end{align*} 
\end{proof}

\begin{theorem}
    We have $$N_m\leq \left(\frac{15}{16}\right)^\frac{m}{r+1}4^m.$$
\end{theorem}

\begin{proof}
    Assume $m=t(r+1)+a$, where $0\leq a<r+1$. Then for any $0\leq s\leq 2^m-1$, we have
    \[s=s_1\times 2^{m-(r+1)}+s_2\times 2^{m-2(r+1)}+\cdots+s_t\times 2^{m-t(r+1)}+s_{t+1},\]
    where $0\leq s_1,s_2,\ldots,s_t\leq 2^{r+1}-1$ and $0\leq s_{t+1}\leq 2^a-1$, which implies an expansion
    \[B_s=\sum_{i=1}^tB_{s_i}\times2^{m-i(r+1)}+B_{s_{t+1}},\]
    and thus
    \[|B_s|\leq \prod_{i=1}^t|B_{s_i}|\cdot |B_{s_{t+1}}|.\]

    Applying the above inequality, we obtain
\begin{align*}
N_m=&\sum_{s=0}^{2^m-1}|B_s| \\
\leq& \sum_{s_1,\ldots,s_t=0}^{2^{r+1}-1}\sum_{s_{t+1}=0}^{2^a-1}\left(\prod_{i=1}^t|B_{s_i}|\cdot |B_{s_{t+1}}|\right) \\
=&\prod_{i=1}^t\left(\sum_{s_i=0}^{2^{r+1}-1}|B_{s_i}|\right)\cdot\sum_{s_{t+1}=0}^{2^a-1}|B_{s_{t+1}}| \\
=&N_{r+1}^tN_a\le(15\times4^{r-1})^t4^a \\
=&\left(\frac{15}{16}\right)^t 4^m\leq \left(\frac{15}{16}\right)^{\frac{m}{r+1}} 4^m.
\end{align*} 
The proof is now complete.
\end{proof}

\begin{corollary}
The rate $R(D_m)$ converges to $0$; so the generalized BCH family $F_n$ (with $n=2^r+1$) is of unit rate.
\end{corollary}

\begin{proof}
    We denote the rate of $D_m$ by $R_m$. Then $R_m={\rm rank}(D_m)/4^m\leq N_m/4^m$ and thus
    \[R_m\leq \left(\frac{15}{16}\right)^{\frac{m}{r+1}}.\]
    Hence $\{R_m\}$ converges to $0$ as $m$ goes to infinity.
\end{proof}

\section{Some cases of three-bit $n$ by using a computer}

\begin{definition}
    Let $A,B$ be two matrices, say $A=(a_{i_1,j_1})_{m_1\times n_1},B=(b_{i_2,j_2})_{m_2\times n_2}$. Then the tensor product of two matrices is $A\otimes B:=(a_{i_1,j_1}B)_{m_1\times n_1}$, namely each block entry is the product of the matrix $B$ and an entry of $A$.
\end{definition}

Note that the entry in the $((i_1-1)m_1+i_2)^{th}$ row and the $((j_1-1)n_1+j_2)^{th}$ column of $A\otimes B$ is $a_{i_1,j_1}b_{i_2,j_2}$.

\begin{definition}
    Let $A,B$ be two matrices of the same size, say $A=(a_{i_1,j_1})_{m\times n},B=(b_{i_2,j_2})_{m\times n}$. Then the Hadamard product of two matrices is $A \circ B:=(a_{i,j}b_{i,j})_{m\times n}$.
\end{definition}

\begin{theorem}[\cite{Horn1985MatrixA}]
\label{tensor_product}
    Let $A,B$ be two matrices. Then 
    \[\mathrm{rank}(A\otimes B)=\mathrm{rank}(A)\mathrm{rank}(B).\]
\end{theorem}

\begin{corollary}
\label{Hadamard_product}
     Let $A,B$ be two $m\times n$ matrices. Then 
     \[\mathrm{rank}(A\circ B)\leq \mathrm{rank}(A)\mathrm{rank}(B).\]
\end{corollary}

\begin{proof}
    By Theorem \ref{tensor_product}, it suffices to show that $A\circ B$ is a submatrix of $A\otimes B$. Let $R=\{(i-1)m+i\mid i=1,2,\ldots,m\},L=\{(j-1)m+j\mid j=1,2,\ldots,n\}$ and $C=A\otimes B(L\times R)$. Then we have
    \begin{align*}
        C(i,j)&=A\otimes B((i-1)m+i,(j-1)n+j) \\
        &=a_{i,j}b_{i,j}=A\circ B(i,j).
    \end{align*}
    The proof is now complete.
\end{proof}

As used previously, $(h(x_1,x_2,y_1,y_2))_{(x_1,x_2),(y_1,y_2)\in \bb F_q^2}$ denotes the matrix in which each $((x_1,x_2), (y_1,y_2))$-entry is $h(x_1,x_2,y_1,y_2)$. When there is no ambiguity regarding the finite field $\mathbb{F}_q$, we can omit the subscript and simply write the matrix as $(h(x_1,x_2,y_1,y_2))$ or $(h)$. Below we always assume that the finite field has characteristic $2$. We have the following lemma.

\begin{lemma}
    Let $i$ be a non-negative integer. Then 
        \[\mathrm{rank}((h(x_1,x_2,y_1,y_2)))=\mathrm{rank}((h(x_1,x_2,y_1,y_2)^{2^i})).\]
\end{lemma}

\begin{proof}
    Note that $h(x_1,x_2,y_1,y_2)^{2^i}=h(x_1^{2^i},x_2^{2^i},y_1^{2^i},y_2^{2^i})$. Furthermore, this expression represents a permutation of both the rows and columns. Thus the result follows.
\end{proof}

\begin{proposition}
\label{inequality_rank_polynomial}
    Let $d(x_1,x_2,y_1,y_2)\in \mathbb{F}_q[x_1,x_2,y_1,y_2]$ and $t$ a positive integer. Then for any integer $m>t$, we have 
    $$\mathrm{rank}((d^{2^m-1}))\leq c\cdot \left(\mathrm{rank}((d^{2^t-1}))\right)^{ \frac{m}{t}},$$
    where $c=\max\{\mathrm{rank}((d^{2^i-1}))\mid 0\leq i<t\}$ only depends on $t$.
\end{proposition}

\begin{proof}
    Suppose that $m=k\cdot t+r$, where $0\leq r< t$. Then we can write 
    \begin{align*}
    2^m-1=\sum_{j=0}^{k-1}{(2^t-1)2^{j+r}}+2^r-1.    
    \end{align*} 
    Hence,
    \begin{align*}
        &\mathrm{rank}((d^{2^m-1}))\\
        \leq& \mathrm{rank}((d^{2^r-1}))\prod_{j=0}^{k-1}\mathrm{rank}((d^{(2^t-1)\times2^{j+r}})) \\      \leq&c\cdot\prod_{j=0}^{k-1}\mathrm{rank}((d^{2^t-1}))=c\cdot \left(\mathrm{rank}((d^{2^t-1}))\right)^k \\
        \leq& c\cdot \left(\mathrm{rank}((d^{2^t-1}))\right)^{ \frac{m}{t}},
    \end{align*}
    where $c=\max\{\mathrm{rank}((d^{2^i-1}))\mid 0\leq i<t\}$.
\end{proof}

\begin{remark}
     The above theorem tells us that the rank of $A_m=(d^{2^t-1})_{\bb F_{2^m}^2 \times \bb F_{2^m}^2}$ will give an upper bound for the rank of $(d^{2^m-1})$. However, the matrix $A_m$ is changing as $m$ increases. We next show that $\mathrm{rank}(A_m)$ would not change when $m$ is sufficiently large.
\end{remark}

\begin{definition}[Rank of  a polynomial]
    Assume that a polynomial $h\in \F[x_1,x_2,y_1,y_2]$, say $$h=\sum_{i_1,i_2,j_1,j_2}a_{i_1,i_2,j_1,j_2}x_1^{i_1}x_2^{i_2}y_1^{j_1}y_2^{j_2},$$ 
    where $a_{i_1,i_2,j_1,j_2}\in \F$. Then the coefficient matrix of $h$, the rows indexed by $(i_1,i_2)$ and the columns indexed by $(j_1,j_2)$, is $M_h=(a_{i_1,i_2,j_1,j_2})$. The {\it rank of the polynomial $h$} is the rank of its coefficient matrix $\mathrm{rank}(M_h)$, and it will be denoted by $\mathrm{rank}(h)$.
\end{definition}

\begin{lemma}
\label{equaility_rank_polynomial}
    Let $h\in \F[x_1,x_2,y_1,y_2]$. Assume that $d=\max\{\deg_{x_1}h,\deg_{x_2}h,\deg_{y_1}h,\deg_{y_2}h)\}$, where $\deg_{x_1}h$ is the degree of $h$ in variable $x_1$. If $q>d$, then $$\mathrm{rank}((h)_{\bb F_q^2\times \bb F_q^2})=\mathrm{rank}(h).$$ 
\end{lemma}

\begin{proof}
    Suppose that $$h=\sum_{i_1,i_2,j_1,j_2=0}^{q-1}a_{i_1,i_2,j_1,j_2}x_1^{i_1}x_2^{i_2}y_1^{j_1}y_2^{j_2}.$$ Then we have
    \begin{align*}
        &(h)_{\bb F_q^2\times \bb F_q^2}=LM_hR\\
        =&\left[\begin{array}{ccc}
        \cdots&x_1^{i_1}x_2^{i_2}&\cdots   
        \end{array}\right]
        \left[\begin{array}{ccc}
             & \vdots& \\
             \cdots&a_{i_1,i_2,j_1,j_2}&\cdots \\
             &\vdots&
        \end{array}\right]
        \left[\begin{array}{c}
         \vdots  \\
          y_1^{j_1}y_2^{j_2} \\
          \vdots
    \end{array}\right],
    \end{align*}
    where the rows of $ L$ and the columns of $ R$ are indexed by elements of $\bb F_q^2$. Note that the matrices $L,R$ are invertible, so $\mathrm{rank}((h))=\mathrm{rank}(M_h)=\mathrm{rank}(h)$.
\end{proof}

We now consider the generalized BCH family $F_n$. In the following, we denote $d=(x_1+y_1)^n+x_2+y_2,f=d^{q-1}$.

\begin{theorem}
    If there exists a positive integer $t$ such that
    \[\mathrm{rank}(d^{2^t-1})< 4^t,\]
    then the generalized BCH family $F_n$ is of unit rate.
\end{theorem}

\begin{proof}
    When $m$ satisfies $2^m>n(2^t-1)$, by Proposition $\ref{inequality_rank_polynomial}$ and Lemma $\ref{equaility_rank_polynomial}$ we have 
    \begin{align*}
        \mathrm{rank}((f)_{\bb F_{2^m}^2\times \bb F_{2^m}^2})&=
        \mathrm{rank}((d^{2^m-1})_{\bb F_{2^m}^2\times \bb F_{2^m}^2}) \\
        &\leq c\cdot \left(\mathrm{rank}((d^{2^t-1})_{\bb F_{2^m}^2\times \bb F_{2^m}^2})\right)^\frac{m}{t} \\
        &= c\cdot \left(\mathrm{rank}(d^{2^t-1})\right)^\frac{m}{t} \\
        &\leq c\cdot (4^t-1)^\frac{m}{t}\\
        \frac{\mathrm{rank}((f))}{4^m}&\leq \frac{c\cdot (4^t-1)^\frac{m}{t}}{4^m}
        =c\cdot \left(\frac{4^t-1}{4^t}\right)^\frac{m}{t}.
    \end{align*}
    Hence the evaluation matrix $(f)_{\bb F_q^2\times \bb F_q^2}$ is of low rank and thus the parity-check matrix of the generalized BCH family $F_n$ is also of low rank. We are done.
\end{proof}
    
We can use a computer to search for the smallest $t$ such that the rank of the polynomial $d^{2^t-1}$ is strictly smaller than $4^t$. For instance, using Magma, we know that $\mathrm{rank}(d^{2^6-1})=3256<4096=4^6$ for $F_7$, $\mathrm{rank}(d^{2^7-1})=15018<16384=4^7$ for $F_{11}$, and $\mathrm{rank}(d^{2^7-1})=14442<16384=4^7$ for $F_{13}$. Therefore, we obtain the following result.

\begin{corollary}
The generalized BCH families $F_7,F_{11}$ and $F_{13}$ are all of unit rate. $\hfill\square$
\end{corollary}


\section*{Acknowledgments}

We would like to thank Sihuang Hu for bringing the open problem of Barg and Z\'emor to our attention. We also thank Zhen Jia for his help with computer programming. 

\end{document}